\documentclass[runningheads]{llncs}

\usepackage{graphicx}
\usepackage{amsmath,amssymb}
\usepackage{textcomp}
\usepackage{comment}
\usepackage{color}
\usepackage[utf8]{inputenc}

\newcommand{\mymatrix}[2]{\left( \begin{array}{#1} #2 \end{array} \right)}
\newcommand{\myvector}[1]{\mymatrix{c}{#1}}
\newcommand{\myvectorr}[1]{\mymatrix{r}{#1}}
\newcommand{\mypar}[1]{\left( #1 \right)}

\newcommand{\cent}[0]{\mbox{\textcent}}
\newcommand{\dollar}{\$}

\newcommand{\AMAfA}{\mathsf{AM(AfA)}}
\newcommand{\AMAfAZ}{\mathsf{AM_{\mathbb{Z}}(AfA)}}
\newcommand{\AMAfAQ}{\mathsf{AM_{\mathbb{Q}}(AfA)}}
\newcommand{\AM}[1]{\mathsf{AM(#1)}}
\newcommand{\UAM}[1]{\mathsf{UAM(#1)}}
\newcommand{\AMQ}[1]{\mathsf{AM_{\mathbb{Q}}(#1)}}
\newcommand{\AMZ}[1]{\mathsf{AM_{\mathbb{Z}}(#1)}}
\newcommand{\UAMQ}[1]{\mathsf{UAM_{\mathbb{Q}}(#1)}}
\newcommand{\NP}{\mathsf{NP}}

\newcommand{\REG}{\mathsf{REG}}
\newcommand{\LOG}{\mathsf{L}}
\newcommand{\DP}{\mathsf{P}}
\newcommand{\NEXP}{\mathsf{NEXP}}
\newcommand{\PSPACE}{\mathsf{PSPACE}}
\newcommand{\SPACE}[1]{\mathsf{SPACE(#1)}}
\newcommand{\ALL}{\mathsf{ALL}}
\newcommand{\UALL}{\mathsf{UALL}}

\newcommand{\subsetsum}{\mathtt{SUBSETSUM}}
\newcommand{\usquare}{\mathtt{USQUARE}}
\newcommand{\upoly}{\mathtt{UPOLY(P)}}
\newcommand{\pal}{\mathtt{PAL}}
\newcommand{\upal}{\mathtt{UPAL}}
\newcommand{\twin}{\mathtt{TWIN}}

\begin{document}

\title{Affine automata verifiers}
\titlerunning{Affine automata verifiers}

\author{Aliya Khadieva\inst{1,2} \and Abuzer Yakary{\i}lmaz\inst{1,3}\orcidID{0000-0002-2372-252X} }
\authorrunning{A. Khadieva and A. Yakary{\i}lmaz}

\institute{
University of Latvia, Riga, Latvia \and 
Kazan Federal University, Kazan, Russia \and 
QWorld Association, \url{https://qworld.net} \\
\email{aliya.khadi@gmail.com~~~~abuzer@lu.lv} 
}

\maketitle

\begin{abstract}
We initiate the study of the verification power of AfAs as part of Arthur-Merlin (AM) proof systems. We show that every unary language is verified by a real-valued AfA verifier. Then, we focus on the verifiers restricted to have only integer-valued or rational-valued transitions. We observe that rational-valued verifiers can be simulated by integer-valued verifiers, and, their protocols can be simulated in nondeterministic polynomial time. We show that this bound tight by presenting an AfA verifier for NP-complete problem SUBSETSUM. We also show that AfAs can verify certain non-affine and non-stochastic unary languages.
\end{abstract}

\section{Introduction}
Affine finite automata (AfAs) are quantum-like generalization of probabilistic finite automata (PFAs) mimicking quantum interference and having the capability of ``making measurement'' based on $\ell_1$-norm (called weighting). The computation of an AfA is linear, but the weighting operators may be non-linear.

AfAs was formally defined in \cite{DCY16}, and it was shown that they are more powerful than PFAs and quantum finite automata (QFAs) in bounded-error and unbounded-error settings, but their nondeterministic version is equivalent to nondeterministic QFAs. Since then, AfAs and their different generalizations (e.g., OBDDs and using counters) have been investigated in a series of work \cite{VY16,HMY17,NKPVY17,IKPY18,VY18,HMY19,HMY21}.

In this paper, we initiate the study of the verification power of AfAs as part of Arthur-Merlin (AM) proof systems. We show that every unary language is verified by a real-valued AfA verifier. Then, we focus on the verifiers restricted to have only integer-valued or rational-valued transitions. We observe that rational-valued verifiers can be simulated by integer-valued verifiers, and, their protocols can be simulated in nondeterministic polynomial time. We show that this bound tight by presenting an AfA verifier for NP-complete problem $\subsetsum$. We also show that AfAs can verify certain non-affine and non-stochastic unary languages. In our protocols, we use similar verification strategies and encoding techniques previously used for two-way QFAs in \cite{YS10B,Yak13C,SY17}. 

In the rest of this section, we provide a quick literature review related to our results. We give the notation and definitions in Section~\ref{sec:pre}, and we review some basic computation techniques by integer-valued affine states and operators used in our proofs in Section~\ref{sec:encoding}. We present our main result on unary languages in Section~\ref{sec:verify-unary}, and our results on rational- or integer-valued AfAs are given in Section~\ref{sec:integer-AM-AfA}. We close the paper by summarizing our results with previously known related results based on complexity classes in Section~\ref{sec:summary}.

\subsection{The computational power of AfAs compared to PFAs and QFAs}
\label{sec:computation-AfA-PFA-QFA}

We review the previously known results comparing AfAs with PFAs and QFAs. 

Bounded-error PFAs and QFAs recognize all and only regular languages \cite{Rab63,KW97,LQZLWM12,AY15}. But, bounded-error AfAs can recognize some nonregular languages such as $ \upal = \{a^nb^n \mid n > 0 \} $ and $ \pal = \{ w \in \{a,b\}^* \mid w=w^r \} $ \cite{DCY16}. Moreover, AfAs can be very succinct compared to PFAs and QFAs \cite{VY16,VY18}, i.e., they can recognize a family of regular languages with bounded-error by using only two states, but the number of states of bounded-error PFAs or QFAs cannot be bounded for this family. 

The class of languages recognized by PFAs with cutpoints is called stochastic languages \cite{Rab63}. QFAs recognize all and only stochastic languages with cutpoints \cite{YS09C,YS11A}. Similar to bounded-error case, AfAs are more powerful than both, and they can recognize some non-stochastic languages \cite{DCY16}. On the other hand, in the nondeterministic setting (when the cutpoint is fixed to zero), QFAs and AfAs have the same computational power \cite{DCY16}.

Regarding the limitations on the computational power of AfAs, we know that \cite{Tur81,YS10A,VY18,HMY21}:
\begin{itemize}
    \item (one-sided or two-sided) bounded-error rational-valued and integer-valued AfAs have the same computational power;
    \item one-sided bounded-error rational-valued AfAs cannot recognize any nonregular unary language;
    \item algebraic-valued AfAs cannot recognize certain non-stochastic unary languages in $ \LOG $ even with unbounded-error (with cutpoints); and,
    \item the class of languages recognized by bounded-error rational-valued AfAs is a proper subset of $ \LOG $.
\end{itemize}

One open problem is whether bounded-error rational-valued AfAs can recognize any nonregular unary language, and, one untouched direction is the computational capabilities of real-valued AfAs.

\subsection{The verification power of PFAs and QFAs}
\label{sec:verification-PFA-QFA}

Interactive proof systems (IPSs) \cite{GMR89} with PFA verifiers \cite{DS92} can verify some nonregular languages such as $ \twin = \{ wcw \mid w \in \{a,b\}^* \} $ with bounded error. The same result is valid for IPS with QFA verifiers communicating with the prover classically.\footnote{When the proof system is fully quantum, we know little \cite{NY09}: the restricted QFA model defined in \cite{KW97} can verify only regular languages with bounded-error.} IPSs are also called private-coin systems since a verifier can hide its probabilistic decisions from the prover. In this way, the verifier can use stronger verification strategies as a part of the protocol (between the verifier and prover) since the prover may not guess the exact configuration of the verifier, and so it may not easily mislead the verifier when it is not honest.

When the computation of a verifier is fully seen by the prover, the system is called public-coin or AM system \cite{Bab85,Con93}. AM system with PFA verifiers \cite{CHPW98} cannot recognize any nonregular languages with bounded error, and we do not know whether AM systems with QFA verifiers can recognize any nonregular language with bounded-error. 

When considering the known results for AfAs (Section~\ref{sec:computation-AfA-PFA-QFA}), there are two natural questions about the verification power of AM systems with rational-valued AfA verifiers:
\begin{enumerate}
    \item whether we can go beyond $ \LOG $ and, if so, how far, and,
    \item whether some nonregular unary languages can be verified or not.
\end{enumerate}
We answer both questions positively, and we obtain $ \NP $ as the tight upper bound for non-unary languages.

\subsection{Two-way PFAs and QFAs}
\label{sec:two-way-PFA-QFA}

As mentioned above, AfAs can recognize nonregular languages $ \upal $ and $ \pal $ with bounded-error without interacting with any prover. Similar results can be obtained for PFAs and QFAs when reading the input many times by using a two-way head \cite{Fre81,AW02}. We review basic facts about bounded-error two-way PFAs and QFAs to have a better picture for our results on AfAs. 

The language $ \upal $ is recognized by bounded-error two-way QFAs \cite{AW02} in polynomial expected time and as well as by two-way PFAs \cite{Fre81} but only in exponential expected time \cite{GW86}. 

The language $ \pal $ can be recognized by bounded-error two-way QFAs in exponential expected time \cite{AW02}, but it cannot be recognized in polynomial expected time even if two-way QFAs augmented with logarithmic amount of space \cite{Rem21}. On the other hand, AM systems with two-way PFA verifiers cannot verify $ \pal $ with bounded error even if augmented with logarithmic space \cite{DS92}. Besides, two-way bounded-error PFAs can recognize only regular languages in polynomial expected time \cite{DS90}, and it is open whether AM systems with two-way PFAs can verify any nonregular languages in polynomial time. 

Regarding unary languages, bounded-error two-way PFAs cannot recognize any nonregular language \cite{Kan91B}, and it is open whether any unary nonregular language is verified by a bounded-error AM system with two-way PFA verifier \cite{CHPW98}. It is also open whether bounded-error two-way QFAs can recognize any nonregular unary language. 

The class of languages verified by AM systems with two-way rational-valued PFA verifiers is a proper subset of $ \DP $ \cite{Con93}. Therefore, the verification power of AfAs can go beyond the verification power of two-way PFAs.

On the other hand, AM systems with two-way QFAs are very powerful \cite{Yak13C,SY17}. Two-way QFAs can verify every unary language in exponential expected time, and so their their verification power is equivalent to that of AfAs on unary languages. On non-unary languages, rational-valued two-way QFAs can verify every language in $\PSPACE$ and some NEXP-complete languages. Therefore, AM systems with rational-valued AfAs are weaker than AM systems with rational-valued two-way QFAs. Here, we should note that AfA verifiers read the input once, but two-way QFAs may run in exponential or double-exponential expected time.

 \section{Preliminaries}
 \label{sec:pre} 
 
Throughout the paper, $ |\cdot|  $ refers
to the $\ell_1$-norm; $ \Sigma $ denotes the input alphabet not containing symbols $\cent$ and $\mathdollar$, respectively called the left and right end-markers; $ \Tilde{\Sigma} $ is the set $ \Sigma \cup \{ \cent,\mathdollar \} $; $ \Sigma^* $ denotes the set of all strings defined on the alphabet $ \Sigma $ including the empty string denoted $ \varepsilon $; and, for a given string $ w \in \Sigma^* $, $ \tilde{w} $ denotes the string $ \cent w \mathdollar $. Moreover, for any string $w$, $ |w| $ is the length of $ w $, $ |w|_\sigma $ is the number of occurrences of symbol $ \sigma $ in $w$, and, whenever $ |w|>0 $, $ w_i $ represents the $ i $-th symbol of $ w $, where $ 1 \leq i \leq |w| $. For an automaton $ M $, $ f_M(w) $ represents the accepting probability of $ M $ on the input $w \in \Sigma^*$. 

A \textit{realtime} automaton reads the given input symbol by symbol and from the left to the right. On each symbol, a realtime automaton can stay a fixed amount of steps. If there is no waiting steps, then it is called \textit{strict realtime}. In this paper, we focus on only the strict realtime models. For every given input $ w $, it is fed as $ \tilde{w} $ so that the automaton can make pre-processing and post-processing while reading the symbols $ \cent $ and $ \dollar $.

An $ m $-state affine system is represented by $ \mathbb{R}^m $, and affine state of this system is represented by $ m $-dimensional vector:
$$
    v = \myvector{\alpha_1 \\ \vdots \\ \alpha_m } \in \mathbb{R}^m 
$$
satisfying that $ \sum_{j=1}^n \alpha_j =1 $, where $\alpha_j$, similar to the amplitudes in quantum systems, is the value of the system being in state $ e_j $.

Any affine operator of this system is a linear operator represented by an $ (m \times m) $-dimensional matrix:
$$
    A = \mymatrix{ccc}{ a_{1,1} &  \cdots & a_{1,m} \\
    \vdots & \ddots & \vdots \\ a_{m,1} & \cdots & a_{m,m} } \in \mathbb{R}^{m \times m}
$$
satisfying that $ \sum_{j=1}^n a_{j,i} =1 $ for each column $i$ (the column summation is 1). When the operator $ M $ is applied to the affine state $ v $, the new state is $ v' = M \cdot v $.

To retrieve information from the affine system, similar to the measurement operators of quantum system, we apply a weighting operator. When the affine state $v$ is weighted, the $i$-th state is observed with probability
$$
    \frac{|\alpha_i|}{|v|} =
    \frac{|\alpha_i|}{|\alpha_1|+\cdots+|\alpha_m|}.
$$

If the system is restricted to have only the non-negative real numbers, then it turns out to be a probabilistic system.

\subsection{Finite automata with deterministic and affine states}

Similar to finite automata with quantum and classical states (QCFA) \cite{AW02}, a finite automaton with deterministic and affine states (ADfA) is an $ n $-state deterministic finite automaton having an $ m $-state affine register, where $ m,n > 0 $. Let $ S = \{s_1,\ldots,s_n\} $ be the deterministic states and $ E=\{e_1,\ldots,e_m\} $ be the affine states, where $ e_i $ is the standard basis in $ \mathbb{R}^m $ with all zeros except the $i$-th entry which is 1.

The computation is governed classically. During the computation, each transition of an ADfA has two parts: affine and classical parts.
\begin{enumerate}
    \item Affine transition: For each pair of deterministic state and reading symbol, say $ (s,\sigma) $, either an affine operator or a weighting operator is applied to the affine register. 
    \item Classical transition can be two types:
    \begin{enumerate}
        \item If an affine operator is applied, then the next classical state is determined based on $ (s,\sigma) $.
        \item If a weighting operator is applied, then the next classical state is determined based  on $ (s,\sigma,e) $, where $ e \in E $ is the measured affine state. 
    \end{enumerate}
\end{enumerate}  

In this paper, we apply the weighting operator only after reading the whole input, and so, we keep the formal definition of the models simpler: a single transition updates both the classical and affine parts at the same time.

Formally, a ADfA $ M $ with $n$ classical and $ m $ affine states is a 8-tuple
$$
    M = (S, E, \Sigma, \delta, s_I,e_I,s_a,E_a),
$$
where 
\begin{itemize}
    \item $ S $ and $ E $ are the set of states as specified above;
    \item $ \delta $ is the transition function described below;
    \item $ s_I \in S $ and $ e_I \in E $ are the  deterministic and affine initial states, respectively; and,
    \item $ s_a \in S $ is the deterministic accepting state;
    \item $ E_a \subseteq E $ is the set of affine accepting state(s).
\end{itemize}

Let $  w \in \Sigma^* $ be the given input of length $ l $. The ADfA reads the input as $ \tilde{w} = \cent w \dollar $ from the left to the right and symbol by symbol. The computation of $ M $ is traced by a pair $ (s,v) $ called a configuration, where $ s \in S $ is the classic state and $ v \in \mathbb{R}^{n \times n} $ is the affine state. At the beginning of the computation, $ M $ is in $ (s_I,v_0) $, where the affine state $v_0 = e_I $.

The transition function is defined as $ \delta: S \times \tilde{\Sigma} \rightarrow S \times \mathbb{R}^{m \times m} $. Let $ (s,v_j) $ be the configuration after the $ j $-th step and $ \sigma = \tilde{w}_j \in \tilde{\Sigma}$. Then the new configuration is $ (s',v_{j+1}) $, where $ \delta(s,\sigma) = (s',A) $ and $ v_{j+1} = A  v_j $.

After reading $ \dollar $ symbol, if the final classical state is not $ s_a $, then the input is  rejected deterministically:
\[
    f_M(w) = 0.
\]
Otherwise, a weighting operator is applied and the input is accepted if an affine accepting state is observed. We denote the final state as $ v_f = v_{|\tilde{w}|} $. Then, the accepting probability by the affine part is
\[
    f_M(w) = \frac{ \sum_{e_i \in E_a} |v_f[i]| } {|v_f|} \in [0,1].
\]

We remark that the ADfA $ M $ defined here can be exactly simulated by the original model defined in \cite{DCY16} with $ (m \cdot n) $ affine states.

\subsection{Affine automata verifiers}

In this paper, we study only Arthur-Merlin type of interactive proof systems where the verifiers are affine automata. In \cite{CHPW98}, Arthur-Merlin systems with probabilistic finite automata verifier is   defined as an automata having both nondeterministic and probabilistic states. We follow the same framework here. We indeed give the ability of making nondeterministic transitions to ADfA models.

A finite automaton with nondeterministic and affine states (ANfA) with $n$ classical and $m$ affine states is formally a 8 tuples
$$
    N = (S, E, \Sigma, \delta, s_I,e_I,s_a,E_a),
$$
where all elements are the same as ADfA except the transition function. For the pair $ (s,\sigma) \in S \times \tilde{\Sigma} $, it can have one or more transitions: 
\[
    \delta(s,\sigma) \rightarrow \{ (s'_1,A_1),\ldots,(s'_k,A_k) \},
\] where each pair $ (s,\sigma) $ can have a different $k > 0 $ value. When having more than one transition, $N$ picks each of them nondeterministically by creating a new path. In this way, $ N $ forms a computation tree, where the root is the starting configuration. Remark that the computation in each path is the same as that of ADfAs and each path may have a different accepting probability. Each path here refers to the  communication with a different prover. 

\subsection{Language recognition and verification}

A language $ L \subseteq \Sigma^* $ is said to be recognized by an ADfA $ M $ with error bound $ \epsilon < \frac{1}{2} $, if 
\begin{itemize}
    \item for every $ w \in L $, $ f_M(w) \geq 1-\epsilon $, and
    \item for every $ w \notin L $, $ f_M(w) \leq \epsilon $.
\end{itemize}
Shortly, we can also say that $ L $ is recognized by $ M $ with bounded error or $L$ is recognized by a bounded-error ADfA.

A language $ L \subseteq \Sigma^* $ is said to be verified by an ANfA $ V $ with error bound $ \epsilon < \frac{1}{2} $, if 
\begin{itemize}
    \item for every $ w \in L $, there is path on which $ f_V(w) \geq 1-\epsilon $, and
    \item for every $ w \notin L $, $ f_V(w) \leq \epsilon $ on each path.
\end{itemize}
Shortly, we can also say that $ L $ is verified by $ V $ with bounded error or $L$ is verified by a bounded-error ANfA.

\subsection{Language classes}

We define $ \AMAfA $ as the class of languages verifiable by bounded-error Arthur-Merlin system having  realtime affine finite verifiers. Any language verifiable by a bounded-error ANfA is in this class, and we obtain all results in this paper by ANfAs. Remark that a model of

realtime affine finite verifiers is more general than ANfA as applying weighting operators more than once and the outcomes can also be processed classically. 

If the verifier is a PFA, QFA, two-way PFA, or two-way QFA, then the related class is $ \AM{PFA} $, $ \AM{QFA} $, $\AM{2PFA}$, or $ \AM{2QCFA} $, respectively, where 2QCFA is the two-way QFA model defined in \cite{AW02}.

We denote an AM class where the verifiers are restricted to have rational-valued or integer-valued components by using subscript $ \mathbb{Q} $ or $ \mathbb{Z} $: $ \AMQ{\cdot} $ or $ \AMZ{\cdot} $.

Here is the list of standard complexity classes mentioned in the paper:
\begin{center}
\begin{tabular}{lcl}
    $ \REG $ & : & \mbox{regular languages} \\
    $ \LOG $ & : & \mbox{logarithmic space} \\ 
    $ \DP $ & : & \mbox{polynomial time} \\ 
    $ \NP $ & : & \mbox{nondeterministic polynomial time} \\
    $ \SPACE{n} $ & : & \mbox{linear space} \\
    $ \PSPACE $ & : & \mbox{polynomial space} \\
    $ \NEXP $ & : & \mbox{nondeterministic exponential space} \\
    $ \ALL $ & : & \mbox{all languages} \\
\end{tabular}
\end{center}

Lastly, for a given complexity class $ \mathsf{C} $, $ \mathsf{UC} $ denotes its unary version.

\section{Basic computation with integer-valued operators}
\label{sec:encoding}

In this section, we review some basic computation techniques by integer-valued affine states and operators, which are later used in our proofs. We use the induction to verify the correctness of encoding techniques.

\subsection{Binary encoding}
\label{sec:binary-encoding}

We read a binary string $ w \in \{0,1\}^* $ and encode its numeric value as the value of an affine state.

We use three affine states. We start in the affine state $ v_0 = (1~~0~~0)^T $, and, after reading $ w $, the affine state is set to
\[
    \myvector{1 \\ val(w) \\ -val(w)},
\]
where $ val(w) $ is the numeric value of $ w $ encoded in the value of $ e_2 $. For symbol $ \sigma \in \{0,1\} $, we use the affine operator $ A_\sigma $ as described below:
\[
    A_0 = \mymatrix{rrr}{1 & 0 & ~~0 \\ 0 & 2 & 0 \\ 0 & -1 & 1}
    \mbox{ and }
    A_1 = \mymatrix{rrr}{1 & 0 & ~~0 \\ 1 & 2 & 0 \\ -1 & -1 & 1}.
\]

\textit{Basis step:} If the first symbol is $0$, the new affine state is
\[
    v_1 = A_0 v_0 = \mymatrix{rrr}{1 & 0 & ~~0 \\ 0 & 2 & 0 \\ 0 & -1 & 1} 
    \myvector{1 \\ 0  \\ 0} = \myvector{1 \\ 0 \\ 0},
\]
where the values of $ e_2 $ and $ e_3 $ are 0. If the first symbol is $1$, the new affine state is
\[
    v_1 = A_1 v_0 = \mymatrix{rrr}{1 & 0 & ~~0 \\ 1 & 2 & 0 \\ -1 & -1 & 1} 
    \myvectorr{1 \\ 0  \\ 0} 
    = \myvectorr{1 \\ 1 \\ -1},
\]
where the values of $ e_2 $ and $ e_3 $ are 1 and $-1$, respectively.

\textit{Induction step:} After reading the $ j $-th symbol, we assume that the affine state is 
\[
    v_j = \myvector{1 \\ x \\ -x},
\]
where $ x $ is numeric value of $ w_1 w_2 \cdots w_j $. If the $ (j+1) $-th symbol is 0, the new affine state is
\[
    v_{j+1} = A_0 v_j = \mymatrix{rrr}{1 & 0 & ~~0 \\ 0 & 2 & 0 \\ 0 & -1 & 1}
    \myvectorr{1 \\ x \\ -x}
    =
    \myvector{1 \\ 2x \\ -2x},
\]
where we can observe that $ 2x = val( w_1 \cdots w_j 0) $.
Similarly, if the $ (j+1) $-th symbol is 1, the new affine state is
\[
    v_{j+1} = A_0 v_j = \mymatrix{rrr}{1 & 0 & ~~0 \\ 1 & 2 & 0 \\ -1 & -1 & 1}
    \myvectorr{1 \\ x \\ -x}
    =
    \myvector{1 \\ 2x+1 \\ -2x-1},
\]
where we can observe that $ 2x +1 = val(w_1 \cdots w_j 1) $.

\subsection{Linear counting}
\label{sec:counting-x}

We read the string $ 0^l $ and encode $ l $ as the value of an affine state. We present two different methods.

\textbf{Method 1:} We start in the affine state $ v_0 = (1~~0~~0)^T $, and, for each symbol $ 0 $, we apply the affine operator
\[
    A = \mymatrix{rrr}{1 & 0 & 0 \\ 1 & 1 & 0 \\ -1 & 0 & 1}.
\]
After reading $ l $ symbols, the affine state is 
$$ v_l = \myvectorr{1 \\ l \\ -l}, $$
where $ l $ is encoded in the value of $e_2$.

\textit{Basis step:} After reading one symbol:
\[
    v_1 = A v_0 = \myvectorr{1 \\ 1 \\ -1}, 
\]
where the value of $ e_2 $ is $1$.

\textit{Induction step:} When in $ v_i $, we calculate $ v_{i+1} $:
\[
    v_{i+1} = \mymatrix{rrr}{1 & 0 & 0 \\ 1 & 1 & 0 \\ -1 & 0 & 1}
    \myvectorr{1 \\ i \\ -i} = \myvector{1 \\ i+1 \\ -(i+1)},
\]
where the value of $ e_2 $ is $i+1$.

\textbf{Method 2:} We start in the affine state $ v_0 = \myvector{1 \\ 0} $, and, for each symbol $0$, we apply the following operator:
\[
    B = \mymatrix{rr}{ 0 & ~-1 \\ 1 & 2 }
\]
After reading $ l $ symbols, the affine state is
$$ v_l = \myvector{1-l \\ l }, $$
where $ l $ is encoded in the value of $e_2$.

\textit{Basis step:} After reading one symbol:
\[
    v_1 = B v_0 = \myvectorr{0 \\ 1}, 
\]
where the value of $ e_2 $ is $1$.

\textit{Induction step:} When in $ v_i $, we calculate $ v_{i+1} $:
\[
    v_{i+1} = \mymatrix{rr}{ 0 & ~-1 \\ 1 & 2 }
    \myvector{1-i \\ i} = \myvector{ -i \\ 1+i } = \myvector{  1-(i+1) \\ i+1 },
\]
where the value of $e_2$ is $ i+1 $.

\subsection{Calculating $ x^2 $}
\label{sec:calculating-xx}

We read the string $ 0^l $ and encode $l^2$ as the value of an affine state. This can be done in many different ways. 

One trivial solution is directly using the methods in Section~\ref{sec:counting-x}. For example, we know that
\[
    v_{l} = B v_0 =  \mymatrix{rr}{ 0 & ~-1 \\ 1 & 2 }^l \myvector{1 \\ 0} = \myvector{ 1-l \\ l }.
\]
If we use tensor the affine part with itself, we obtain $ l^2 $ as the value of $e_4$ :
\[
    v'_l = (v_l \otimes v_l) = (B \otimes B)^l (v_0 \otimes v_0) = \myvector{ 1-l \\ l } \otimes \myvector{ 1-l \\ l } = \myvector{ (1-l)^2 \\ (1-l) l \\ (1-l)l \\ l^2 }.
\]
If we use the first method in Section~\ref{sec:counting-x}, then 
the dimension of the new affine vector is 9.

An alternative method is using binomial expansions, i.e., $ (i+1)^2 $ is a linear combination of $ i^2,i,1 $ such that $ (i+1)^2 = i^2+2i+1 $. Thus, by using $ 1 $, $ i $, $ i^2 $, we can calculate $ (i+1)^2 $ by a linear operator.

There are different ways of implementing this idea. The first one is that the affine state is of the form
\[
    v_i = \myvector{1 \\ i \\ i^2 \\ -i -i^2}
\]
after reading $ l $ $0$s, and then the corresponding affine operator (for symbol $0$) is
\[
    \mymatrix{rrrr}{1 & 0 & ~~0 & ~~0 \\ 1 & 1 & 0 & 0 \\ 1 & 2 & 1 & 0 \\ -2 & -2 & 0 & 1}.
\]
Here $ v_0 = (1~~0~~0~~0)^T $, and we can check the induction step as
\[
    \myvector{1 \\ 1+i \\ 1+2i+i^2 \\ -2-2i-i-i^2} =
    \mymatrix{rrrr}{1 & 0 & ~~0 & ~~0 \\ 1 & 1 & 0 & 0 \\ 1 & 2 & 1 & 0 \\ -2 & -2 & 0 & 1}
    \myvector{1 \\ i \\ i^2 \\ -i -i^2},
\]
which is equivalent to
\[
    \myvector{1 \\ i+1 \\ (i+1)^2 \\ (-1-i)+(-1-2i-i^2) } =
    \myvector{1 \\ i+1 \\ (i+1)^2 \\ -(i+1)-(i+1)^2 }.
\]

One may also use the following forms of $ v_i $'s:
\[
    v_i = \myvector{ 1 \\ 2i \\ i^2 \\ -2i -i^2 }
    \mbox{ or }
    v_i = \myvector{1-i-i^2 \\ i \\ i^2 }.
\]
In the latter form, 1 is eliminated, which is always obtained as the summation of the vector. The corresponding affine operators are accordingly:
\[
    \mymatrix{rrrr}{1 & 0 & ~~0 & ~~0 \\ 2 & 1 & 0 & 0 \\ 1 & 1 & 1 & 0 \\ -3 & -1 & 0 & 1}
    \mbox{ or }
    \mymatrix{rrr}{ -1 & -4 & -2 \\ 1 & 2 & 1 \\ 1 & 3 & 2 }.
\]
We present the induction step for the latter form. The initial vector is  $ v_0 = (1~~0~~0)^T $, and the vector after $ (i+1)$-th step is
\[
    v_{i+1} = \mymatrix{rrr}{ -1 & -4 & -2 \\ 1 & 2 & 1 \\ 1 & 3 & 2 } \myvector{1-i-i^2 \\ i \\ i^2 } = \myvector{ -i^2 - 3i - 1 \\ i +1 \\ i^2 + 2i + 1 },\]
which is, after the re-arrangement of the first entry,
\[
    v_{i+1}=\myvector{ 1 - (i+1) - (i+1)^2 \\ i +1 \\ i^2 + 2i + 1 }.
\]

\subsection{Calculating polynomials}
\label{sec:calculating-poly}

Now, we generalize the method given in Section~\ref{sec:calculating-xx} using binomial expansions. Let $ P(x) $ be a polynomial with degree $ d $. Then, the set of variables is $ \{x^0=1,x,x^2,\ldots,x^d\} $. Our aim is to have  $ P(l) $ as the value of an affine state after applying the same affine operator $ l $ times.

We read the string $0^l$ and encode $ P(l) $ as the value of an affine state as follows:
\[
    v_l = \myvector{ 1 \\ l \\ l^2 \\ \vdots \\ l^d \\ P(l) \\ \overline{1} },
\]
where $ \overline{1} $ is a variable making the column sum equal to 1.

We know that $ (i+1)^j $ is a linear combination of $ 1, i, \ldots, i^{j} $, and $ P(l) $ is a linear combination of $ 1,l,\ldots, l^d $. We define the affine operator for symbol $ 0 $ as a combination of two affine operators. The first affine operator updates the first $(d+1)$ entries by using binomial coefficients, and, the second affine operator calculates the value of polynomial by using the coefficients of the polynomial:
\[
    \myvector{ 1 \\ i \\ i^2 \\ \vdots \\ i^d \\ P(i) \\ \overline{1} }
    \rightarrow
    \myvector{ 1 \\ (i+1) \\ (i+1)^2 \\ \vdots \\ (i+1)^d \\ P(i) \\ \overline{1} }
    \rightarrow
    \myvector{ 1 \\ (i+1) \\ (i+1)^2 \\ \vdots \\ (i+1)^d \\ P(i+1) \\ \overline{1} }.
\]

\section{Verification of every unary language}
\label{sec:verify-unary}

Let $ L \subseteq \Sigma^* $ be an arbitrary unary language, where $ \Sigma = \{a\} $. We define a real number to encode the whole membership information of $ L $ as follows:
\[
    \alpha_L = \sum_{i=0}^{\infty}  \frac{b_i}{32^{i+1}} = \frac{b_0}{32}+\frac{b_1}{32^2}+\frac{b_2}{32^3}+\cdots ,
\]
where 
\begin{itemize}
    \item $ b_i = 1 $ if $a^i \in L$ and
    \item $ b_i = 0 $ if $a^i \notin L$.
\end{itemize}
In binary form: $ bin(\alpha_L) = 0.0000b_0 0000b_1 \cdots 0000b_i \cdots $.
Moreover, we define 
\[
    \alpha_L[j] = \frac{b_j}{32} + \frac{b_{j+1}}{32^2} + \frac{b_{j+2}}{32^3} + \cdots ,
\]
where $ j \geq 0 $. 

We observe a few basic facts about $ \alpha_L $ and $ \alpha_L[j] $, which we will use in our proofs.
\begin{enumerate}
    \item For any $ \alpha_L[j] $, there is a unary language $ L^{'} $ such that $ \alpha_L[j] = \alpha_{L^{'}} $.
    \item The values of $ \alpha_L $ and so $ \alpha_L[j] $ are bounded:
    \[  0 \leq \alpha_L \leq \frac{1}{31} \mbox{ and } 0 \leq \alpha_L[j] \leq \frac{1}{31}. \]
    \item The values of $ \alpha_L[j+1] $ and $ \alpha_L[j] $ can be related: 
        \begin{itemize}
            \item If  $ b_j = 0 $, then
            \[ \alpha_L[j+1] = 32 \cdot \alpha_L[j]. \]
            \item If $ b_j = 1 $, then
            \[ \alpha_L[j+1] = 32 \cdot \alpha_L[j] - 1.\]
        \end{itemize}
\end{enumerate}

By using $ \alpha_L $, we design a bounded-error ANfA for language $L$. The main idea of the protocol is that each $ b_i $ is nondeterministically guessed and the verification is done by subtracting the guessed $ b_i $ and the actual value $ b_i $ encoded in $ \alpha_L $. As long as the nondeterministic choices are correct, the result of such subtractions will be zero. Otherwise, it will not be zero, based on which we reject the input. The details are given in the proof below.

\begin{theorem}
 \label{thm:every-unary}
 Every unary language $ L \subseteq \{a\}^* $ is verified by an ANfA $ V $ with error bound $0.155$.
\end{theorem}
\begin{proof}
    The verifier $ V $ has two classical states and three affine states, where $ s_2 $ is the classical accepting state and $ e_1 $ is the only affine accepting state. The initial affine state is $ v_0 = \mypar{1~~0~~0}^T $. 
    
    Let $ w = a^l $ be the given input for $ l \geq 0 $. Until reading $\dollar$, $V$ makes two nondeterministic transitions for each symbol:  for $\tilde{w}_i$ ($ i \in \{1,\ldots,l+1\} $), $ V $ guesses the value of $ b_{i-1} $, say $g_{i-1}$. If $g_{i-1} = 0 $, then classical state is set to $ s_1 $, and if $g_{i-1} = 1 $, then classical state is set to $ s_2 $. The affine operators are described below.
    
    On symbol $ \cent $, a combination of two affine operators is applied. In the first part, the affine state is set as
    \[
        \myvector{1 \\ \alpha_L \\ -\alpha_L} =
        \mymatrix{rrr}{1 & 0 & 0 \\ \alpha_L & ~~1 & 
        ~~0 \\ -\alpha_L & 0 & 1 }
        \myvector{1 \\ 0 \\ 0}.
    \]
    In the second part, the affine operator $ A_{g_0} $ is applied, where
    \[
        A_0 = \mymatrix{rrr}{1 & ~-31 & ~-31 \\ 0 & 32 & 0 \\ 0 & 0 & 32}
        \mbox{ and }
        A_1 = \mymatrix{rrr}{1 & ~-31 & ~-31 \\ -1 & 32 & 0 \\ 1 & 0 & 32}.
    \]
    On each symbol $ a $, the second part for symbol $\cent$ is repeated: the affine operator $ A_{g_i} $ is applied on the path where $ g_i $ is picked.
    
    If $ b_0 $ is guessed correctly, then affine state becomes
    \[
        \myvector{1 \\ \alpha_L[1] \\ - \alpha_L[1]} = A_{b_0} \myvector{1 \\ \alpha_L[0] \\ - \alpha_L[0]} .
    \]
    It is sufficient to check the value of $e_2$:
    \begin{itemize}
        \item If $ b_0 = 0 $, after applying $ A_0 $, the value of second entry becomes $ 32 \cdot \alpha_L[0] $, which is equal to $ \alpha_L[1] $.
        \item If $ b_0 = 1 $, after applying $ A_1 $, the value of second entry becomes $ 32 \cdot \alpha_L[0] - 1 $, which is equal to $ \alpha_L[1] $.
    \end{itemize}
    Similarly, as long as the nondeterministic guesses are correct, the affine part evolves as given below:
    \[
        \myvector{1 \\ \alpha_L[1] \\ - \alpha_L[1]} 
        \xrightarrow{~1^{st}~a~}
        \myvector{1 \\ \alpha_L[2] \\ - \alpha_L[2]} 
        \xrightarrow{~2^{nd}~a~}
        \cdots
        \xrightarrow{~l^{th}~a~}
        \myvector{1 \\ \alpha_L[l+1] \\ - \alpha_L[l+1]}.
    \]
    
    Now, we examine the case in which at least one nondeterministic guess is wrong. Assume that $ g_i \neq b_i $ is the first wrong guess (for symbol $ \tilde{w}_{i+1} $). The value of $ e_2 $ is $ \alpha_L[i] $ before this guess, and it becomes 
    $$ 1+\alpha_L[i+1] ~~\mbox{ or }~~  \alpha_L[i+1] - 1 $$ after the guess. Thus, the absolute value of $e_2$ is bounded below by $ 1- \frac{1}{31} = \frac{30}{31} $, which is at least 30 times greater than any $ \alpha_L[j] $. If there is another symbol $ a $ to be read, then the value of $ e_2 $ is multiplied by 32 followed by subtraction of 0 or -1. That means the integer part of the absolute of new value of $e_2$ becomes greater than 30, and so the absolute value of $ e_2$ is at least 900 times greater than any $ \alpha_L[j] $. For each new symbol of $a$, this factor (i.e., 30 and 900) will be multiplied by 30. 
    
    On symbol $ \dollar $, $ V $ does not change the classical state and applies the following operator to the affine state:
    \[
        A_\dollar(k) = \mymatrix{ccc}{~1~ & ~1-k~ & ~1-k~ \\ 0 & k & 0 \\ 0 & 0 & k},
    \]
    where $ k = \frac{31}{2\sqrt{30}}$, which gives the minimum error when maximizing the accepting probability for members and minimizing the same for the non-members.
    
    If $ w \in L $, the path following the correct  nondeterministic guesses ends in classical state $ s_2 $ and affine state $ ( 1 ~~ k \cdot \alpha_L[l+1] ~~-k \cdot \alpha_L[l+1]  )^T $. Remember that $ 0 \leq \alpha_L[l+1] \leq \frac{1}{31}  $. Thus, the input is accepted with probability 
    \[
        \frac{1}{1+2k\alpha_L[l+1]} \geq \frac{1}{1+\frac{2k}{31}} = \frac{1}{1+\frac{1}{\sqrt{30}}} = \frac{\sqrt{30}}{1+\sqrt{30}} = 1 - \frac{1}{1+\sqrt{30}} > 0.845.
    \]
    
    If $ w \notin L $, then we have different cases. (1) If $ b_l $ is guessed correctly ($g_l=0$), then the input is rejected deterministically. (2) If each guess is correct except $ b_l $ ($g_l=1$), then affine state is  
     \[
        \myvector{1 \\ k ( \alpha_L[l+1] -1) \\ - k ( \alpha_L[l+1] - 1) },
    \]
    and so, the accepting probability is
    \[
        \frac{1}{ 1+2k ( 1 - \alpha_L[l+1]) } \leq \frac{1}{1+2k(\frac{30}{31})} = \frac{1}{1+\sqrt{30}} < 0.155.
    \]
    In other words, the rejecting probability is at least $ 1-0.155 = 0.845 $. (3) If the guess $ g_i $ for $ i < l $ is wrong, then, as we described above, the absolute values of $e_2$ and $e_3$ are at least 30 times bigger than that of the case (2), and so is the rejecting probability.
\qed\end{proof}

When defining $ \alpha_L $, the denominators can be some numbers greater than 32, and, in this way we can obtain better error bounds, i.e., arbitrarily close to 0.

\begin{corollary}
    Every unary language $ L \subseteq \{a\}^* $ is verified by ANfAs with arbitrarily small error bounds.
\end{corollary}

\section{$ \AMAfAZ $}
\label{sec:integer-AM-AfA}

Recently, it was shown \cite{HMY21} that any language recognized by a rational-valued ADfA with error bound $ \epsilon $ is recognized by an integer-valued ADfAs with error bound $\epsilon' $, where $ 0 \leq \epsilon \leq \epsilon' < \frac{1}{2} $. The latter automaton is constructed by modifying the components of the former automaton so that, on the same input, the accepting probability of the latter one can differ insignificantly from the accepting probability of the former one, i.e., the difference is at most $ \epsilon' - \epsilon $.  Thus, on the same input, the accepting probabilities for the same nondeterministic path will differ insignificantly, and so the error bound increases but still less than $ \frac{1}{2} $.

\begin{theorem}
    $ \AMAfAQ = \AMAfAZ $.
\end{theorem}

It is known that $ \mathsf{AM(PFA)} = \REG $ \cite{CHPW98}. We do not know whether $ \mathsf{AM(QFA)} $ contains any non-regular language. On the other hand, ADfAs can recognize some non-regular languages with bounded-error such as $ \pal $ requiring at least logarithmic space for bounded-error probabilistic computation \cite{FK94}. A natural question is whether $ \AMAfA $ goes beyond $ \LOG $.

\begin{theorem}
    $ \AMAfAQ \subseteq \mathsf{NP} \cap \SPACE{n} $.
\end{theorem}
\begin{proof}
    Let $ L \in \AMAfAQ $ be a language. Then, there is an ANfA $ V $ verifying $ L $ with error bound $ \epsilon \in \mathbb{Q} \cap [0,\frac{1}{2}) $. 
    
    The descriptions of $ V $ and the error bound are finite, which can be wired into the description of Turing Machines (TMs). For any given input, the computation on each path of $ V $ can be traced by vector and matrix multiplications. As the length of each sequence is linear, all computation including weighting, calculating the accepting probability, and comparing it with the error bound can be done in polynomial time and linear space (i.e., the size of affine state vector is fixed, the precision of each entry can be at most linear, and each new entry is a linear combination of these entries).
    
    In the case of nondeterministic TM simulation, the TM implements the nondeterministic choices of $ V $ directly. In the case of linear-space TM simulation, the TM use a linear counter to check all nondeterministic strategies one-by-one. Even though the overall simulation runs in exponential expected time, the space usage can be bounded linearly. 
\qed\end{proof}

We show that integer-valued ANfAs can verify some NP-complete problems. For this purpose, we use the following language version of the Knapsack problem (Page 491 of \cite{KT06}): $ \subsetsum $ is the language of strings of the form $ S \# B_1 \# \cdots \# B_k $, where
\begin{itemize}
    \item $ S, B_1,\ldots,B_k \in \{0,1\}^{*} $ are binary numbers and
    \item there exists a subset of $\{B_1,\ldots,B_k\}$ that adds up to precisely $S$, i.e.,  $$ \exists I \subseteq \{ 1,\ldots,k \} \mbox{ such that } S = \sum_{i \in I} B_i. $$
\end{itemize} 
Remark that we do not use any non-negative integer, and it is still NP-Complete.

\begin{theorem}
 $\subsetsum$ is verified by an integer-valued ANfA $V(t)$ such that every member is accepted with probability 1 and every non-member is accepted with probability at most $\frac{1}{2t+1}$ for some $ t \in \mathbb{Z^{+}} $.
\end{theorem}
\begin{proof}
Let $ w \in \Sigma^* $, where $ \Sigma = \{0,1,\#\} $. The verifier $ V(t) $, shortly $V$, classically checks $w$ has at least one $ \# $. Otherwise, the input is rejected deterministically. 

In the remaining part, we assume that $ w $ is of the form $ S \# B_1 \# \cdots \# B_k $ for some $ k>0 $. Remark that the binary value of empty string is zero (whenever $ S = \varepsilon $ or any $ B_i = \varepsilon $). The protocol has the following steps:
\begin{enumerate}
    \item $ V $ starts with encoding $ S $ into the value of affine state $ e_2 $.
    \item $ V $ nodeterministically picks some $ B_i $'s ($ 1 \leq i \leq k $). Such decision is made when reading symbols $ \# $.
    \begin{enumerate}
        \item If $ B_i $ is not picked, then affine state does not changed.
        \item Otherwise, $ V $ encodes $ B_i $ into the value of the affine state $ e_3 $, and then, it is subtracted from the value of $ e_2 $ and the value of $ e_3 $ is set to zero.
    \end{enumerate}
    \item At the end of the computation, the decision is made based on the fact that the value of $ e_2 $ is zero for the members and non-zero integer for the non-members. The error is reduced by using certain tricks before the weighting operator.
\end{enumerate}

The affine part has four states $ \{e_1,\ldots,e_4\} $ and $ e_1 $ is the only accepting state. The initial affine state is $ (1~~0~~0~~0)^T $, and it does not change when reading $ \cent $. For encoding binary string, we use the technique described in Section~\ref{sec:binary-encoding}.
The value of $ S $ is  encoded by using the affine operators $ \{A_\sigma \mid \sigma \in \{0,1\} \} $:
\[
    A_0 = \mymatrix{rrrr}{~1 & 0 & ~~0 & ~~0 \\ 0 & 2 & 0 & 0 \\ 0 & 0 & 1 & 0 \\ 0 & -1 & 0 & 1 }
    \mbox{ and }
    A_1 = \mymatrix{rrrr}{1 & 0 & ~~0 & ~~0 \\ 1 & 2 & 0 & 0 \\ 0 & 0 & 1 & 0 \\ -1 & -1 & 0 & 1 },
\]
where the value of $ e_3 $ is not changed. The value of each picked $ B_i $ is encoded by  the affine operators $ \{A'_\sigma \mid \sigma \in \{0,1\} \} $:
\[
    A'_0 = \mymatrix{rrrr}{~1 & ~~0 & 0 & ~~0 \\  0 & 1 & 0 & 0 \\ 0 & 0 & 2 & 0 \\  0 & 0 & -1 & 1 }
    \mbox{ and }
    A'_1 = \mymatrix{rrrr}{1 & ~~0 & 0 & ~~0 \\ 0 & 1 & 0 & 0 \\ 1 & 0 & 2 & 0 \\ -1 & 0 & -1 & 1 },
\]
where the value of $ e_2 $ is not changed.
With the following operator, the value of $ e_3 $ is subtracted from the value of $ e_2 $ and  set to 0:
\[
    D = \mymatrix{rrrr}{1 & ~~0 & 0 & ~~0 \\ 0 & 1 & -1 & 0 \\ 0 & 0 & 0 & 0 \\ 0 & 0 & 2 & 1}.
\]

For a picked subset $ I \subseteq \{1,\ldots,k\} $, let $ S_I = \sum_{i \in I} B_i $. Before weighting operator, for some $ t \in \mathbb{Z^{+}} $, we apply the following operator to decrease the error bound for the non-members:
\[
    E(t) = \mymatrix{rrrr}{~1 & 0 & ~~~0 &0 \\ 0 & t & 0 & 0 \\ 0 & ~1-t & 1 & ~1-t \\ 0 & 0 & 0 & t}.
\]

On the path where $ I $ is followed, just before applying $ E(t) $, the affine state is  
\[
    \myvector{1 \\ S - S_I \\ 0 \\ S_I -S},
\]
and it is 
\[
    \myvector{1 \\ t(S - S_I) \\ 0 \\ t(S_I -S)}
\]
after applying $ E(t) $. It is easy to see that if $ S=S_I $, then the final affine state is $ (1~~0~~0~~0)^T $ and so the input is accepted with probability 1. If $ S \neq S_I $, then $ |S-S_I| \in \mathbb{Z^{+}} $, and so the values of $ e_2 $ and $ e_4  $ are not zero and the accepting probability can be at most $ \dfrac{1}{2t+1} $. 

Therefore, if $ w \in \subsetsum $, then there exists a subset $ I $ satisfying the membership condition and it is picked on a path where the input is accepted with probability 1. If $ w \notin \subsetsum $,  there is no subset satisfying the membership condition, and so the input is accepted with probability at most $ \dfrac{1}{2t+1} $ in each path. The error bound can be  arbitrarily small when $ t \rightarrow \infty $. 
\qed\end{proof}

It is not known that whether there is a NP-Complete unary language or not. It was shown that if there is such a language, then $ \mathsf{P} = \mathsf{NP} $ \cite{Ber78}. Regarding the verification power of rational-valued ANfAs, we use some non-stochastic unary languages.

For a given non-linear polynomial with non-negative integer coefficients $ P(x) $, we define a unary language as $ \upoly = \{ a^{P(i)} \mid i \in \mathbb{N} \} $. Turakainen \cite{Tur81} showed that such languages are not stochastic. Recently, it was shown that \cite{HMY21} they are not algebraic affine languages, too, i.e., they cannot be recognized by algebraic-valued ADfAs with cutpoints. 

Now, we show that ANfAs can verify any $ \upoly $ language with bounded error. We start with a very simple case: $ \usquare = \{a^{i^2} \mid i \in \mathbb{N} \} $. 

\begin{theorem}
    \label{thm:usquare}
    Language $ \usquare$ is verified by an ANfA $V(t)$ with any error bound $\frac{1}{2t+1}$, where $ t \in \mathbb{Z^+} $.
\end{theorem}
\begin{proof}
    We use the parameter $ t $ at the end of the proof, and we represent $ V(t) $ shortly as $ V $. The verifier $V$ uses $ 4 $ affine states, and $ e_1 $ is the single accepting affine state. Let $w=0^l$ be the given input. If $ w = \varepsilon $, then it is accepted classically. We assume that $ w \neq \varepsilon $ in the rest of the proof. 
    
    The protocol of $ V $ is as follows: $ V $ nondeterministically picks a positive integer $ j \geq l $ and then checks whether $ j^2 = l $. If $ w \in \usquare $, then there exists such $ j = \sqrt{l} $ and so this comparison is made successfully in one of the nondeterministic paths. If $ w \notin \usquare $, there is no such $ j $ and so there is no successful comparison in any nondeterministic path.
    
    The verifier follows $ (l+1) $ different paths during its computation:
    \[ 
        path_0,path_1,\ldots,path_l,
    \]
    where the main one is $ path_0 $. We use the encoding techniques given in Section~\ref{sec:counting-x} and \ref{sec:calculating-xx}. When reading the $ i $-th symbol of $w$, $ path_0 $ continues with $ path_0 $ or creates $ path_i $.
    
    On $ path_0 $, $ V $ is in the following affine states after reading $ w_i $ and $ w_l $:
    \[
        v_{0,i+1}=\myvector{1 \\ i \\ i^2 \\ \overline{1}}
        \mbox{ and }
        v_{0,l+1}=\myvector{1 \\ l \\ l^2 \\ \overline{1}},
    \]
    respectively.
    After reading $ w_i $, $ V $ creates $ path_i $, on which it is in the  affine state
    \[
        v_{i,i+1}=\myvector{1 \\ i \\ i^2 \\ \overline{1}}.
    \]
    For the rest of the computation, $ V $ continues with counting the number of symbols on $ e_2 $ but it does not change the value of $ e_3 $ until reading $ \dollar $. The affine state on $path_i$ ($i>0$) after reading $ w_l $ is
    \[
        v_{i,l+1}=\myvector{1 \\ l \\ i^2 \\ \overline{1}}.
    \]
    
    On $path_0$, the input rejected is classically. On $ path_i $, after reading $\dollar$, $ V $ enters the classical accepting state, and it sets the affine state as
    \[
        \myvector{1 \\ t(l-i^2) \\ t(i^2-l) \\ 0  }.
    \]
    
    If $ w \in L$, then on $path_{\sqrt{l}}$, the final affine state is $ e_1 $ and so $w$ is accepted with probability 1.
    
    If $ w \notin L $, then on $ path_i $, the absolute value of $ e_2 $ or $ e_3 $ is $ |t(l-i^2)| $, which is at least $ t $. Thus, the input is accepted with probability at most $ \epsilon = \frac{1}{2t+1} \leq \frac{1}{3} $. It is clear that $ \epsilon \rightarrow 0 $ when $ t \rightarrow \infty $.
\qed\end{proof}

\begin{theorem}
    Language $ \upoly $ is verified by an ANfA $V(t)$ with any error bound $\frac{1}{2t+1}$, where $ t \in \mathbb{Z^+} $.
\end{theorem}
\begin{proof}
    The proof is identical to the proof of Theorem~\ref{thm:usquare} after modify the encoding part (we use the techniques in Section~\ref{sec:counting-x} and \ref{sec:calculating-poly}). First note that $ P(i) \geq i $ since the coefficients of $P$ are non-negative. So, for any $ 0^l \in \upoly $, there exists $ j \leq l $ such that $ l = P(j) $. Second, on $ path_i $, $ P(i) $ is calculated and then the verifier checks whether $ P(i) = l $ or not.
    
    If $0^l$ is in $ \upoly $, then it is accepted with probability 1 in one of the nondeterministic path. If it is not in $ \upoly $, then the accepting probability on any path can be at most $ \frac{1}{2t+1} $.
\qed\end{proof}

\section{Summary}
\label{sec:summary}

On unary languages, for the real-valued verifiers, we show that AfAs and 2QCFAs have the same verification power:
\[
   \UALL = \UAM{2QCFA} =  \UAM{AfA},
\]
where AfA verifiers are realtime machines but 2QCFAs run in exponential expected time. 

On unary languages, for the rational-valued verifiers, we know that
\[
 \mathsf{UREG} = \UAMQ{PFA} \subseteq
    \begin{array}{l}
         \UAMQ{QFA} \subseteq \UAM{QFA}  \\
         \UAMQ{2PFA} \subseteq \UAM{2PFA}
    \end{array},
\]
where it is open if the inclusions are strict, and we show that $ \upoly \in \UAMQ{AfA} $ and so we have
\[
    \mathsf{UREG} \subsetneq \UAMQ{AfA}.
\]

On non-unary languages, for the rational-valued verifiers, we give an upper bound for $ \AMAfAQ $, and so we have
\[
    \AMAfAQ = \AMAfAZ \subseteq \NP \cap \SPACE{n} \subsetneq \AMQ{2QCFA},
\]
where 2QCFAs run in double-exponential expected time. Our bound is tight since we show that
\[
    \subsetsum \in \AMAfAZ.
\]

\section*{Acknowledgements}
Yakary{\i}lmaz was partially supported by the ERDF project Nr. 1.1.1.5/19/A/005 ``Quantum computers with constant memory''. 

A part of research is funded by the subsidy allocated to Kazan Federal University for the state assignment in the sphere of scientific activities, project No. 0671-2020-0065.

\bibliographystyle{splncs04}
\bibliography{tcs}

\end{document}